\documentclass[aps,pra,groupedaddress,tightenlines]{revtex4}

\usepackage{amsmath}
\usepackage{graphicx}
\usepackage[english]{babel}
\usepackage{hyperref}

\usepackage{graphicx}
\usepackage{subfigure}  
\usepackage{epsfig}
\usepackage[english]{babel}
\usepackage{amssymb}
\usepackage{amsfonts}
\usepackage{amsmath}
\usepackage{amsthm}
\usepackage{amscd}  
\usepackage{paralist} 
\usepackage{psfrag}

\newtheorem{theorem}{Theorem} 
\newtheorem{lemma}{Lemma} 
\newtheorem{proposition}{Proposition} 

\begin{document}

\title{Fault tolerance for holonomic quantum computation\footnote{A chapter in the book \textit{Quantum Error Correction}, edited by Daniel A. Lidar and Todd A. Brun, (Cambridge University Press, 2013), 
\href{http://www.cambridge.org/us/academic/subjects/physics/quantum-physics-quantum-information-and-quantum-computation/quantum-error-correction}{http://www.cambridge.org/us/academic/subjects/physics/quantum-physics-quantum-information-and-quantum-computation/quantum-error-correction}.}}

\author{Ognyan Oreshkov}
\affiliation{QuIC, Ecole Polytechnique, CP 165, Universit\'{e} Libre de Bruxelles, 1050 Brussels, Belgium}

\author{Todd A. Brun}
\affiliation{Departments of Electrical Engineering and Physics, and Center for Quantum Information Science \& Technology, University of Southern California, Los Angeles, California
90089, USA}

\author{Daniel A. Lidar}
\affiliation{Departments of Chemistry, Electrical Engineering, and Physics, and Center for Quantum Information Science \& Technology, University of Southern California, Los Angeles, California
90089, USA}

\begin{abstract}
We review an approach to fault-tolerant holonomic quantum computation on stabilizer codes. We explain its workings as based on adiabatic dragging of the subsystem containing the logical information around suitable loops along which the information remains protected.
\end{abstract}

\maketitle


\section{Introduction}\label{introduction}

In Chapter 16 it was shown how holonomic quantum computation
(HQC) can be combined with the method of decoherence-free subspaces
(DFSs), leading to passive protection against certain types of
correlated errors. However, this is not enough for fault tolerance
since other types of errors can accumulate detrimentally unless
corrected. Scalability of HQC therefore requires going beyond that
scheme, e.g., by combining the holonomic approach with \emph{active}
error correction. One way of combining HQC with active quantum
error-correcting codes, which is similar to the way HQC is combined
with\index{DFS (decoherence-free subspace)} DFSs, was also mentioned in Chapter 16. This approach,
however, is not scalable since it requires Hamiltonians that commute
with the stabilizer elements, and when the code increases in size,
this necessitates couplings that become increasingly nonlocal.

In this chapter, we will show how HQC can be made fault tolerant by
combining it with the techniques for fault-tolerant quantum error
correction (FTQEC) on stabilizer codes using Hamiltonians of finite
locality. The fact that the holonomic method can be mapped directly
to the circuit model allows us to construct procedures which resort
almost entirely to these techniques. We will discuss an approach
which makes use of the encoding already present in
a stabilizer code and does not require additional qubits
\cite{Oreshkov:0812.4682, Oreshkov:2008:070502}. An alternative approach which requires ancillary qubits\index{ancillary qubits} for implementing\index{transversal!operation} transversal operations between qubits in the code can be found in Ref.~\cite{Oreshkov:2009:090502}.

Since protected information is contained in subsystems
\cite{Knill:2006:042301, Blume-Kohout:2008:030501} (see
Chapter 6), the problem of implementing fault-tolerant HQC
can be understood as that of manipulating fault-tolerantly the
subsystem containing the protected information by holonomic means.
We therefore begin by first introducing a generalization of the
standard HQC method, which is applicable to encoding in subsystems.

\section{Holonomic quantum computation on
subsystems}\label{e4:HQCsubsystems}

As pointed out in Chapter 6, protected quantum information is
most generally contained in the subsystems $\mathcal{H}^A_i$ in a
decomposition of the Hilbert space of the system of the form
\begin{equation}
\mathcal{H}^S=\bigoplus_{i=1}^m\mathcal{H}^A_i\otimes\mathcal{H}^B_i \oplus \mathcal{K}.\label{e4:decompositionfull}
\end{equation}
Here the dimensions of the subsystems are related by
$\textrm{dim}\mathcal{H}^S=\overset{m}{\underset{i=1}{\sum}}
\textrm{dim}\mathcal{H}^A_i \times\textrm{dim}\mathcal{H}^B_i + \textrm{dim}\mathcal{K}$. In
the formalism of operator quantum error correction\index{quantum error correction!operator}
\cite{Kribs:2005:180501, Kribs:2006:382, Poulin:2005:230504,
Oreshkov:2008:022333}, the subsystems $\mathcal{H}^A_i$ contain the logical information, and $\mathcal{H}^B_i$ contain the syndrome and gauge degrees of freedom. In the most general case,
all subsystems $\mathcal{H}^A_i$ can be used for encoding and
computation \cite{Beny:2007:100502}.
By using adiabatic holonomies it is possible to apply arbitrary computations
in the subsystems $\mathcal{H}^A_i$, a result which is summarized in
the following theorem \cite{Oreshkov:2009:090502}:

\begin{theorem}\label{e4:Theorem1sub} Consider a non-trivial decomposition into
subsystems of the form \eqref{e4:decompositionfull}. Choose an initial
Hamiltonian in the form
\begin{gather}
H(0)=\bigoplus_{i=1}^m I^A_i\otimes H^B_i \oplus H_{\mathcal{K}},\label{e4:H0HQSsub}
\end{gather}
where $H^B_i$ are operators on $\mathcal{H}^B_i$ such that all
eigenvalues of $H^B_i$ are different from the eigenvalues of $H_{\mathcal{K}}$ and the eigenvalues of $H^B_j$
for $i\neq j$. In the case when $\textrm{dim}\mathcal{K}=0$ and $m=1$, we impose the additional
requirement that $H^B_1$ has at least two different eigenvalues. By
varying this Hamiltonian adiabatically along suitable loops in a
sufficiently large control manifold, it is possible to generate a
unitary of the form
\begin{gather}
U=\bigoplus_{i=1}^m W_i^A\otimes V^B_i \oplus V_{\mathcal{K}},
\end{gather}
where $\{W_i^A\}$ is any desired set of geometric unitary transformations
on $\{\mathcal{H}^A_i\}$.
\end{theorem}
The proof of the theorem is based on the following lemma.

\begin{lemma}\label{e4:Lemma1sub} By varying a Hamiltonian adiabatically along
suitable loops in a sufficiently large control manifold, it is
possible to implement holonomically any combination of unitary
transformations in its eigenspaces.
\end{lemma}

\textit{Comment}. As discussed in Chapter 16, if we have
sufficient control over the parameters of a Hamiltonian, we can
generate holonomically any unitary operation in a given eigenspace.
This lemma concerns the question of whether it is possible to
generate any \textit{combination} of holonomies in the different
eigenspaces. Although intuitively expected based on considerations
concerning the generic irreducibility of the adiabatic connection,\index{connection!adiabatic}
the property may not be obvious. For example, in the case of a
two-level Hamiltonian, the evolution of one of the eigenspaces
completely determines the evolution of the other one, which raises
the question of whether it is possible to obtain independent
holonomies in the two eigenspaces. We now show that this is
possible. Note that even though the proof is constructive and can
serve as a general prescription for simultaneous HQC in different
eigenspaces, it is primarily meant as a proof of principle.

\begin{proof}[Proof of Lemma~\ref{e4:Lemma1sub}]
It is sufficient to show
that it is possible to generate an arbitrary operation in any given
eigenspace while at the same time generating the identity operation
in the rest of the eigenspaces. Without loss of generality, we will
assume that there are only two eigenspaces; if there are more, we
can operate within the subspace spanned by two of them at a time, by
varying only the restriction of the Hamiltonian on that subspace.
Then the initial Hamiltonian can be written as
\begin{gather}
H(0)=\varepsilon_1\Pi_1+\varepsilon_2\Pi_2,
\end{gather}
where $\Pi_{1,2}$ are the projectors\index{projector} on the ground and excited
eigenspaces, and $\varepsilon_1<\varepsilon_2$ are their
corresponding eigenvalues. Notice that this Hamiltonian is invariant
under unitary transformations of the form
\begin{gather}
V=V_1\oplus V_2,\label{e4:V}
\end{gather}
where $V_{1,2}$ are unitaries on the subspaces with projectors
$\Pi_{1,2}$, respectively. Let the Hamiltonian vary along a loop
$H(t)$, $H(0)=H(T)$, which satisfies the adiabatic requirement to
some satisfactory precision. To this precision, the resulting
unitary transformation can be written as
\begin{gather}
U(T)=\mathcal{T}\textrm{exp}(-i\int_0^TdtH(t))=e^{-i\omega_1}U_1
\oplus e^{-i\omega_2}U_2,
\end{gather}
where $U_{1}$ and $U_2$ are the holonomies resulting in the two
eigenspaces, and $\omega_{1,2}=\int_0^T dt\varepsilon_{1,2}(t)$ are
dynamical phases. Observe that the Hamiltonian $VH(t)V^{\dagger}$,
where $V=V_1\oplus V_2$, is a valid loop based on $H(0)$ with the
same spectrum as that of $H(t)$, which gives rise to the holonomies
$V_1U_1V_1^{\dagger}$ and $V_2U_2V_2^{\dagger}$, respectively. This
follows from the fact that the overall unitary transformation
generated by $VH(t)V^{\dagger}$ is equal to $VU(t)V^{\dagger}$ where
$U(t)$ is the unitary generated by $H(t)$.

Imagine that we want to generate holonomically the unitary
transformation $W_1$ in the ground space of the Hamiltonian while at
the same time obtaining the identity\index{holonomy} holonomy $I_2$ in the excited
space. Choose any loop $H(t)$ which gives rise to the\index{holonomy} holonomy
$W_1^{\frac{1}{d_2}}$ in the ground space, where $d_2$ is the
dimension of the excited space (we know that such a loop can be
found). Let this loop result in the\index{holonomy} holonomy $W_2$ in the excited
space. The latter can be written as
$W_2=\overset{d_2}{\underset{j=1}{\sum}}
e^{i\alpha_j}|j\rangle\langle j|$, where $\{|j\rangle\}$ is an
eigenbasis of $W_2$ and $e^{i\alpha_j}, \alpha_j\in R$, are the
corresponding eigenvalues. Consider the unitary transformation $C_2$
which cyclicly permutes the eigenvectors $\{|j\rangle\}$:
$C_2|j\rangle=|j+1\rangle$, where we define $|d_2+1\rangle\equiv
|1\rangle$. We can implement the desired combination of holonomies
in the two eigenspaces as follows. First apply $H(t)$. This results
in the holonomies $W_1^{\frac{1}{d_2}}$ and $W_2$ in the ground and
excited spaces, respectively. Next, apply $(I_1\oplus
C_2)H(t)(I_1\oplus C_2)^{\dagger}$. This generates the holonomies
$W_1^{\frac{1}{d_2}}$ and
$C_2W_2C_2^{\dagger}=\overset{d_2}{\underset{j=1}{\sum}}
e^{i\alpha_{j-1}}|j\rangle\langle j|$ (we have defined
$\alpha_{1-1}\equiv \alpha_{d_2}$). The combined effect of these two
operations is $W_1^{\frac{2}{d_2}}$ and
$\overset{d_2}{\underset{j=1}{\sum}}
e^{i(\alpha_j+\alpha_{j-1})}|j\rangle\langle j| $. We next apply
$(I_1\oplus C_2^2)H(t)(I_1\oplus C_2^2)^{\dagger}$, which generates
the holonomies $W_1^{\frac{1}{d_2}}$ and
$C_2^2W_2C_2^{2\dagger}=\overset{d_2}{\underset{j=1}{\sum}}
e^{i\alpha_{j-2}}|j\rangle\langle j|$. The net result becomes
$W_1^{\frac{3}{d_2}}$ and $\overset{d_2}{\underset{j=1}{\sum}}
e^{i(\alpha_j+\alpha_{j-1}+\alpha_{j-2})}|j\rangle\langle j| $. We
continue this for a total of $d_2$ rounds, which results in the net
holonomic transformations $W_1^{\frac{d_2}{d_2}}=W_1$ and
$e^{i(\alpha_1+\alpha_2+...+\alpha_{d_2})}
\overset{d_2}{\underset{j=1}{\sum}} |j\rangle\langle j|\propto I_2$.
\end{proof}

\begin{proof}[Proof of Theorem~\ref{e4:Theorem1sub}] For the purposes of this theorem, $\mathcal{K}$ can be regarded as a particular $\mathcal{H}^B_i$ whose co-factor $\mathcal{H}^A_i$ is 1-dimensional, so we can assume that $\textrm{dim}\mathcal{K} = 0$ for simplicity. Let us denote the
eigenvalues of $H^B_i$ in Eq.~\eqref{e4:H0HQSsub} by
$\omega_{\alpha_i}$ where $\alpha_i=1,2,...,d_i$, and the projectors
on their corresponding eigenspaces $\mathcal{H}^B_{\alpha_i}$ by
$\Pi^B_{\alpha_i}$. Then the spectral decomposition\index{spectral decomposition} of the initial
Hamiltonian reads
$H(0)=\sum_{i=1}^m\sum_{\alpha_i=1}^{d_i}\omega_{\alpha_i}
\Pi^A_{i}\otimes \Pi^B_{\alpha_i}$, where $\Pi^A_{i}$ is the
projector on $\mathcal{H}^A_i$. According to Lemma~\ref{e4:Lemma1sub},
we can implement holonomically any combination of unitary
transformations in the different eigenspaces of $H(0)$ up to an
overall phase. If we want to implement the set of unitary operations
$\{W^A_i\}$ in the different subsystems $\mathcal{H}^A_i$, we can do
this by implementing the\index{holonomy} holonomy $W^A_i\otimes W^B_{\alpha_i}$ in
each of the eigenspaces
$\mathcal{H}^A_i\otimes\mathcal{H}^B_{\alpha_i}$ for
$\alpha_i=1,2,...,d_i$ where $W^B_{\alpha_i}$ are arbitrary
unitaries on $\mathcal{H}^B_{\alpha_i}$. This results in the net
unitary
\begin{equation}
U=\bigoplus_iW_i^A\otimes
(\bigoplus_{\alpha_i}e^{i\phi_{\alpha_i}}W^B_{\alpha_i})\equiv
\bigoplus_iW_i^A\otimes V^B_i,
\end{equation}
where $e^{i\phi_{\alpha_i}}$ are dynamical phases resulting in the
eigenspaces $\mathcal{H}^A_i\otimes\mathcal{H}^B_{\alpha_i}$. 
\end{proof}

To summarize, HQC on a subsystem can be realized by adiabatically
varying a Hamiltonian that acts locally on the corresponding
co-subsystem. During the evolution, the information initially
encoded in the subsystem transforms to a different subsystem which
is related to the initial one via a geometric unitary operation. The
dynamical part of the unitary factors out as a transformation on the
correspondingly transformed co-subsystem. The problem of FTHQC can
be understood as that of finding a fault-tolerant realization of
this approach.

\section{FTHQC on stabilizer codes without additional qubits}
\label{e4:originalscheme}

\subsection{The main idea}

The developed techniques for FTQEC on stabilizer codes provide a
prescription for how to encode information and how to transform the
subsystem containing the information (hereafter referred to as the code subsystem) so that the class of errors for
which the code is designed remains correctable. We will try to find
realizations of the same transformations that the code subsystem
follows in a standard fault-tolerant scheme using the generalized
method of HQC on subsystems. There are two difficulties in this
respect that have to be considered. First, not every unitary
evolution of a subsystem can be realized by holonomic means. For
example, a general evolution inside a fixed subsystem cannot be
implemented holonomically because the HQC method requires that the
encoded states leave the original code in order to undergo
non-trivial geometric transformations. Second, the holonomic
approach unavoidably gives rise to dynamical transformations on the
co-subsystem (see below), and these could jeopardize the fault
tolerance of the scheme. We will see that neither of these features
is a fundamental obstacle to the realization of fault-tolerant HQC.

The standard fault-tolerant techniques are primarily based on the
use of\index{transversal!operation} transversal operations. In addition, there are
non-transversal\index{transversal!operation} operations for preparing and verifying a special
ancillary state such as Shor's ``cat'' state
$(|00...0\rangle+|11...1\rangle)/\sqrt{2}$. Since single-qubit
unitaries together with the CNOT gate form a universal set\index{universal!set of quantum gates} of
gates, fault-tolerant computation can be realized entirely in terms
of single-qubit gates and transversal CNOT gates,\index{transversal!gate!CNOT} assuming that the
ancillary state can be prepared reliably. Thus we can aim at
constructing holonomic gates on the code subsystem via
transformations that in the original basis of the full Hilbert space
are equivalent to transversal one- and two-qubit gates or to
operations for the preparation of the ``cat'' state.

It turns out, however, that holonomic transformations on the code
subsystem cannot be realized using purely\index{transversal!operation} transversal operations
without the use of extra qubits, even if the encoded gate has a
purely transversal implementation in the dynamical case. This is
because the Hamiltonians that leave the code subsystem invariant are
linear combinations of elements of the stabilizer or, more generally (in the case of operator codes), the gauge group
of the code, and they unavoidably couple qubits in the same block.
But\index{transversal!operation} transversal operations are not the most general class of
operations that do not lead to propagation of errors. A
transformation which at every moment is equal to a transversal
operation followed by a syndrome-preserving transformation on the
co-factor of the subsystem that contains the protected information
is also fault-tolerant. In fact, any\index{transversal!operation} transversal operation in a
given fault-tolerant protocol can be safely replaced by an
operation of the latter type. It is this latter type of
transformations by means of which we can realize fault-tolerant HQC
without the use of extra qubits.

We will show that by choosing as a starting Hamiltonian a suitable
element of the stabilizer or the gauge group of the code and varying
this Hamiltonian along appropriately chosen paths in parameter
space, we can generate operations that -- from the perspective of
the full Hilbert space -- transform both the ground and the excited
spaces via the same\index{transversal!operation} transversal operation. This means, in
particular, that the transformation of the code subsystem is the
same as the one that would result from the application of the
\index{transversal!operation} transversal operation, up to a transformation on the co-subsystem.
The relative dynamical phase that accumulates between the ground and
excited spaces is equal to a phase on the co-subsystem which is
irrelevant for the fault-tolerance of the scheme since it is either
projected out when a measurement of the syndrome is performed, or is
equivalent to a gauge transformation. Thus by an appropriate
sequence of such transformations we can take the code subsystem
adiabatically along a loop such that the resultant\index{holonomy} holonomy at the
end is by construction equal to a given encoded gate. This is the
main idea behind our scheme. We note, however, that our scheme does
not use an isodegenerate Hamiltonian along a given loop. This is
because for simplicity we use Hamiltonians that are equal to a
single element of the stabilizer or the gauge group at a given time,
and we change the Hamiltonians along different portions of a loop.
In this respect, the scheme we will present is slightly more general
than the one described in Sec.~\ref{e4:HQCsubsystems}, but the main
idea behind its workings is the same---the full Hilbert space splits
into different subspaces, each of which is adiabatically transported
around a loop such that all subspaces undergo the same\index{holonomy} holonomy
which factors out as an operation on the code subsystem. In the case
of standard stabilizer codes, the subspaces in question are the
syndrome subspaces. In the case of subsystem codes, each syndrome
subspace further splits into subspaces containing redundant
information. In the latter case, if along the loop we change between
Hamiltonians that are non-commuting elements of the gauge group,
these redundant subspaces may seem to undergo dynamical
transformations in addition to the geometric ones. However, these
dynamical effects are equivalent to gauge transformations and do not
affect the\index{holonomy} holonomy taking place inside the code subsystem.

For concreteness, we will consider an $[[n,1,r,3]]$ stabilizer code.
This is a code that encodes $1$ qubit into $n$, has $r$ gauge
qubits, and can correct arbitrary single-qubit errors. As we saw in
the previous section, in order to apply holonomic transformations on
the subsystem that contains the logical information, we need a
nontrivial starting Hamiltonian that leaves this subsystem
invariant. It is easy to verify that the only Hamiltonians that
satisfy this property are linear combinations of the elements of the
stabilizer or the gauge group.

Note that the stabilizer and the gauge group transform during the
course of the computation under the operations being applied. At any
stage when we complete an encoded operation, they return to their
initial forms modulo a gauge transformation. Our scheme will follow the same transversal
operations as those used in a standard dynamical fault-tolerant
scheme, but as explained above, in addition we
will have extra dynamical phases that multiply each syndrome
subspace or are equivalent to more general gauge transformations. It is
easy to see that these phases leave the stabilizer invariant and transorm the gauge group via a gauge transformation, so without loss of generality we can omit them from our analysis of the transformation of these groups. During the implementation of
a standard encoded gate, the Pauli group\index{Pauli group} ${G}_n$ on a given codeword
may spread over other codewords, but it can be verified that this
spreading can be limited to at most $4$ other codewords including
the ``cat'' state. This is because the encoded CNOT gate can be
implemented fault-tolerantly on any stabilizer code by a transversal
operation on $4$ encoded qubits \cite{Gottesman:9705052}, and any
encoded Clifford gate can be realized using only the encoded CNOT,
provided that we are able to do fault-tolerant measurements\index{measurement!fault-tolerant} (the
encoded Clifford group\index{Clifford!group}\index{group!Clifford}
 is generated by the encoded Hadamard, Phase
and CNOT gates). Encoded gates outside of the Clifford group, such
as the encoded $\pi/8$\index{quantum gate!$\pi/8$} or\index{fault-tolerant operation!Toffoli gate} Toffoli gates, can be implemented
fault-tolerantly using encoded CNOT gates\index{quantum gate!CNOT} conditioned on the qubits
in a ``cat'' state, so they may require\index{transversal!operation} transversal operations on a
total of $5$ blocks. For CSS codes, however, the spreading of the
Pauli group of one block during the implementation of a basic
encoded operation can be limited to a total of $3$ blocks, since the
encoded CNOT gate has a transversal\index{transversal!gate!CNOT} implementation
\cite{Gottesman:9705052}.

It also should be noted that fault-tolerant encoded Clifford
operations can be implemented using only Clifford gates on the
physical qubits \cite{Gottesman:9705052}. These operations transform
the stabilizer and the gauge group into subgroups of the Pauli
group, and their elements remain in the form of tensor products of
Pauli matrices. The fault-tolerant implementation of encoded gates
outside of the Clifford group, however, involves operations that
take these groups outside of the Pauli group. We will, therefore,
consider separately two cases: encoded operations in the Clifford
group, and encoded operations outside of the Clifford group.

\subsection{Encoded operations in the Clifford group}

\subsubsection{Single-qubit unitary operations}\label{e4:sectionSingleQubit}

For applying transformations on a given qubit, say, the first one,
we will use as a starting Hamiltonian an element of the stabilizer
(with a minus sign) or the gauge group of the code, that acts
non-trivially on that qubit. Since we are considering codes that can
correct arbitrary single-qubit errors, one can always find an
element of the initial stabilizer or the initial gauge group that
has a factor $\sigma_0=I$, $\sigma_1=X$, $\sigma_2=Y$ or
$\sigma_3=Z$ acting on the first qubit, i.e.,
\begin{equation}
\widehat{G}=\sigma_i\otimes \widetilde{G},\hspace{0.4cm}i=0,1,2,3
\label{e4:stabelementgen}
\end{equation}
where $\widetilde{G}$ is a tensor product of Pauli matrices and the
identity on the remaining $n-1$ qubits. It can be verified that
under Clifford gates the stabilizer and the gauge group transform in
such a way that this is always the case except that the factor
$\widetilde{G}$ may spread to qubits in other blocks. We can assume
that the stabilizer spreads on at most 5 blocks including the
``cat'' state, since this is sufficient to implement any encoded
operation. Henceforth, we will use ``hat'' to denote operators on
all qubits on which the stabilizer spreads, and ``tilde'' to denote
operators on all of these qubits except the first one.

Without loss of generality we will assume that the chosen stabilizer
or gauge-group element for that qubit has the form
\begin{equation}
\widehat{G}=Z\otimes \widetilde{G}.\label{e4:stabelement}
\end{equation}
As initial Hamiltonian, we will take the operator
\begin{equation}
\widehat{H}(0)=-\widehat{G}=-Z\otimes \widetilde{G}.\label{e4:H(0)}
\end{equation}

\begin{proposition}\label{e4:Proposition1sub} If the initial Hamiltonian \eqref{e4:H(0)} is
varied adiabatically so that only the factor acting on the first
qubit changes,
\begin{equation}
\widehat{H}(t)=-H(t)\otimes \widetilde{G},\label{e4:Ham1}
\end{equation}
where
\begin{equation}
\textrm{Tr}\{H(t)\}=0,
\end{equation}
the transformation that each of the eigenspaces of this Hamiltonian
undergoes will be equivalent to that driven by a local unitary on
the first qubit, $\widehat{U}(t)\approx U(t)\otimes \widetilde{I}$.
\end{proposition}

\begin{proof}
Observe that \eqref{e4:Ham1} can be written as
\begin{equation}
\widehat{H}(t)=H(t)\otimes \widetilde{P}_0- H(t)\otimes
\widetilde{P}_1,\label{e4:Ham2}
\end{equation}
where
\begin{equation}
\widetilde{P}_{0}=\frac{\widetilde{I}- \widetilde{G}}{2},
\hspace{0.4cm}\widetilde{P}_{1}=\frac{\widetilde{I}+
\widetilde{G}}{2},\label{e4:Projectors}
\end{equation}
are orthogonal complementary projectors. The evolution driven by
$\widehat{H}(t)$ is therefore
\begin{equation}
\widehat{U}(t) = U_0(t)\otimes \widetilde{P}_0 + U_1(t) \otimes
\widetilde{P}_1,\label{e4:overallunitary}
\end{equation}
where
\begin{equation}
U_{0,1}(t)=\mathcal{T}\textrm{exp}(-i\overset{t}{\underset{0}{\int}}\pm
H(\tau)d\tau).\label{e4:unitaries}
\end{equation}

Let $|\phi_{0}(t)\rangle$ and $|\phi_{1}(t)\rangle$ be the
instantaneous ground and excited states of $H(t)$ with eigenvalues
$E_{0,1}(t)=\mp E(t)$ ($E(t)>0$). Then in the adiabatic limit we have
\begin{equation}
U_{0,1}(t)= e^{i\omega(t)}U_{A_{0,1}}(t)\oplus
e^{-i\omega(t)}U_{A_{1,0}}(t),\label{e4:unita}
\end{equation}
where $\omega(t)= \int_0^t d\tau E(\tau)$ and
\begin{gather}
U_{A_{0,1}}(t)=e^{\int_0^t d\tau
\langle\phi_{0,1}(\tau)|\frac{d}{d\tau}|\phi_{0,1}(\tau)\rangle
}|\phi_{0,1}(t)\rangle \langle \phi_{0,1}(0)|.\label{e4:holon}
\end{gather}

The projectors on the ground and excited eigenspaces of
$\widehat{H}(0)$ are
\begin{equation}
\widehat{P}_0=|\phi_0(0)\rangle\langle\phi_0(0)|\otimes
\widetilde{P}_0 + |\phi_1(0)\rangle\langle\phi_1(0)|\otimes
\widetilde{P}_1
\end{equation}
and
\begin{equation}
\widehat{P}_1=|\phi_1(0)\rangle\langle\phi_1(0)|\otimes
\widetilde{P}_0 + |\phi_0(0)\rangle\langle\phi_0(0)|\otimes
\widetilde{P}_1,
\end{equation}
respectively. Using Eq.~\eqref{e4:unita} and Eq.~\eqref{e4:holon}, one
can see that the effect of the unitary \eqref{e4:overallunitary} on
each of these projectors\index{projector} is
\begin{equation}
\widehat{U}(t)\widehat{P}_0=e^{i\omega(t)}(U_{A_{0}}(t)\oplus
U_{A_{1}}(t))\otimes\widetilde{I} \hspace{0.1cm}\widehat{P}_0,
\end{equation}
\begin{equation}
\widehat{U}(t)\widehat{P}_1=e^{-i\omega(t)}(U_{A_{0}}(t)\oplus
U_{A_{1}}(t))\otimes\widetilde{I} \hspace{0.1cm}\widehat{P}_1,
\end{equation}
i.e, up to an overall dynamical phase its effect on each of the
eigenspaces is the same as that of the unitary
\begin{equation}
\widehat{U}(t)=U(t)\otimes\widetilde{I},
\end{equation}
where
\begin{equation}
U(t)=U_{A_{0}}(t)\oplus U_{A_{1}}(t)\label{e4:finalU}.
\end{equation}
\end{proof}

We next show how by suitably choosing $H(t)$ we can implement all
necessary single-qubit gates. We will identify a set of points in
parameter space, such that by interpolating between these points we
can draw various paths resulting in the desired transformations.

Consider the single-qubit unitary operator
\begin{equation}
V^{\theta\pm}=\frac{1}{\sqrt{2}}\begin{pmatrix} 1& \mp e^{-i\theta}\\
\pm e^{i\theta}&1
\end{pmatrix},
\end{equation}
where $\theta$ is a real parameter. Note that $V^{\theta
\mp}=(V^{\theta \pm})^{\dagger}$. Define the following single-qubit
Hamiltonian:
\begin{equation}
H^{\theta\pm}\equiv V^{\theta\pm}ZV^{\theta\mp}.
\end{equation}
Let $H(t)$ in Eq.~\eqref{e4:Ham1} be a Hamiltonian\index{Hamiltonian!interpolating} that interpolates
between $H(0)=Z$ and $H(T)=H^{\theta\pm}$ (up to a factor) as
follows:
\begin{equation}
H(t)=f(t)Z+g(t) H^{\theta\pm}\equiv H^{\theta\pm}_{f,g}(t),
\label{e4:interpolation}
\end{equation}
where $f(0),g(T)>0$, $f(T)=g(0)=0$. To simplify our notation, we
will drop the indices $f$ and $g$ of the Hamiltonian, since the
exact form of these functions is not important for our analysis as
long as they are sufficiently smooth (see discussion in
Sec.~\ref{e4:sectionAdicond}). This Hamiltonian has eigenvalues $\pm
\sqrt{f(t)^2+g(t)^2}$ and its energy gap is non-zero unless the
entire Hamiltonian vanishes. It can be shown that in the adiabatic
limit, the Hamiltonian \eqref{e4:Ham1} with $H(t)=H^{\theta\pm}(t)$
gives rise to the effective transformation
\begin{equation}
\widehat{U}^{\theta\pm}(T)=V^{\theta\pm}\otimes \widetilde{I}
\end{equation}
on each eigenspace. Details of the proof can be found in
Ref.~\cite{Oreshkov:2008:022325}.

We will use this result to construct a set of standard gates by
sequences of operations of the form $V^{\theta\pm}$, which can be
generated by interpolations of the type \eqref{e4:interpolation} run
forward or backward. For single-qubit gates in the Clifford group,
we will only need three values of $\theta$: $0$, $\pi/2$ and
$\pi/4$. For completeness, however, we will also demonstrate how to
implement the $\pi/8$ gate, which together with the Hadamard gate\index{quantum gate!Hadamard} is
sufficient to generate any single-qubit unitary transformation
\cite{Boykin:1999:486}. For this we will need $\theta=\pi/8$. Note
that
\begin{equation}
H^{\theta\pm}= \pm (\cos{\theta}X+\sin{\theta}Y),
\end{equation}
so for these values of $\theta$ we have $H^{0\pm}=\pm X$,
$H^{\pi/2\pm}=\pm Y$, $H^{\pi/4\pm}=\pm
(\frac{1}{\sqrt{2}}X+\frac{1}{\sqrt{2}}Y)$, $H^{\pi/8\pm}=\pm
(\cos{\frac{\pi}{8}}X+\sin{\frac{\pi}{8}}Y)$.

Consider the adiabatic interpolations between the following
Hamiltonians:
\begin{equation}
-Z\otimes\widetilde{G} \rightarrow -Y\otimes \widetilde{G}
\rightarrow Z\otimes\widetilde{G}.\label{e4:Xgate}
\end{equation}
According to the above result, the first interpolation yields the
transformation $V^{\pi/2+}$. The second interpolation can be
regarded as the inverse of $Z\otimes\widetilde{G}\rightarrow
-Y\otimes\widetilde{G}$ which is equivalent to
$-Z\otimes\widetilde{G}\rightarrow Y\otimes\widetilde{G}$ since
$\widehat{H}(t)$ and $-\widehat{H}(t)$ yield the same geometric
transformations. Thus the second interpolation results in
$(V^{\pi/2-})^{\dagger}=V^{\pi/2+}$. The net result is therefore
$V^{\pi/2+}V^{\pi/2+}=iX$. We see that up to a global phase\index{global phase} the
above sequence results in an implementation of the $X$ gate in each
eigenspace.

Similarly, one can verify that the $Z$ gate can be realized via
the loop
\begin{equation}
-Z\otimes\widetilde{G} \rightarrow -X\otimes
\widetilde{G}\rightarrow Z\otimes\widetilde{G} \rightarrow
Y\otimes \widetilde{G}\rightarrow
-Z\otimes\widetilde{G}.\label{e4:Zgate}
\end{equation}

The Phase gate $P$ can be realized by applying
\begin{equation}
-Z\otimes\widetilde{G}\rightarrow
-(\frac{1}{\sqrt{2}}X+\frac{1}{\sqrt{2}}Y)\otimes\widetilde{G}\rightarrow
Z\otimes\widetilde{G},
\end{equation}
followed by the $X$ gate.

The Hadamard gate $W$ can be realized by first applying $Z$,
followed by
\begin{equation}
-Z\otimes\widetilde{G} \rightarrow -X\otimes \widetilde{G}.
\end{equation}

Finally, the $\pi/8$ gate $T$ can be implemented by first applying
$Y=iXZ$, followed by
\begin{equation}
Z\otimes\widetilde{G}\rightarrow
-(\cos{\frac{\pi}{8}}X+\sin{\frac{\pi}{8}}Y)\otimes\widetilde{G}\rightarrow
-Z\otimes\widetilde{G}.
\end{equation}

\subsubsection{A note on the adiabatic
condition}\label{e4:sectionAdicond}

Before we show how to implement the CNOT gate, let us comment on the
conditions under which the adiabatic approximation assumed in the
above operations is satisfied. Because of the form
\eqref{e4:overallunitary} of the overall unitary, the adiabatic
approximation depends on the extent to which each of the unitaries
\eqref{e4:unitaries} approximate the expressions \eqref{e4:unita}. The
latter depends only on the adiabatic properties of the
non-degenerate two-level Hamiltonian $H(t)$. For such a Hamiltonian,
the simple version of the adiabatic condition
\cite{Messiah:1965:North} reads
\begin{equation}
\frac{\varepsilon}{{\Delta}^2}\ll 1, \label{e4:adi}
\end{equation}
where
\begin{equation} \varepsilon =
\operatornamewithlimits{max}_{0\leq t\leq T}|\langle
\phi_1(t)|\frac{dH(t)}{dt}|\phi_0(t)\rangle|, \label{e4:eps}
\end{equation}
and
\begin{equation}
\Delta = \operatornamewithlimits{min}_{0\leq t\leq
T}(E_1(t)-E_0(t))= \operatornamewithlimits{min}_{0\leq t\leq
T}2E(t) \label{e4:Delt}
\end{equation}
is the minimum energy gap of $H(t)$.

Along the segments of the parameter paths we described, the
Hamiltonian is of the form \eqref{e4:interpolation} and its derivative
is
\begin{equation}
\frac{dH^{\theta\pm}(t)}{dt}=\frac{df(t)}{dt}Z+\frac{dg(t)}{dt}H^{\theta\pm},
\hspace{0.6cm} 0<t<T.
\end{equation}
This derivative is well defined as long as $\frac{df(t)}{dt}$ and
$\frac{dg(t)}{dt}$ are well defined. The curves we described,
however, may not be differentiable at the points connecting two
segments. In order for the Hamiltonians \eqref{e4:interpolation} that
interpolate between these points to be differentiable, the functions
$f(t)$ and $g(t)$ have to satisfy $\frac{df(T)}{dt}=0$ and
$\frac{dg(0)}{dt}=0$. This means that the change of the Hamiltonian
slows down to zero at the end of each segment (except for a possible
change in its strength), and increases again from zero along the
next segment. We point out that when the Hamiltonian stops changing,
we can turn it off completely by decreasing its strength. This can
be done arbitrarily fast and it would not affect a state which
belongs to an eigenspace of the Hamiltonian. Similarly, we can turn
on another Hamiltonian for the implementation of a different
operation.

The above condition guarantees that the adiabatic approximation is
satisfied with precision
$1-\textit{O}((\frac{\varepsilon}{{\Delta}^2})^2)$. It is known,
however, that under certain conditions on the Hamiltonian, we can
obtain better results \cite{Hagedorn:2002:235}. Let us write the
Schr\"{o}dinger equation as
\begin{equation}
i\frac{d}{dt}|\psi(t)\rangle = H(t) |\psi(t)\rangle\equiv
\frac{1}{\epsilon} \bar{H}(t) |\psi(t)\rangle,
\end{equation}
where $\epsilon>0$ is small. Assume that $\bar{H}(t)$ is smooth and
all its derivatives vanish at the end points $t=0$ and $t=T$ (note
that $\bar{H}(t)$ is non-analytic at these points, unless it is
constant). Then if we keep $\bar{H}(t)$ fixed and vary $\epsilon$,
the adiabatic error would scale super-polynomially with $\epsilon$,
i.e., the error will decrease with $\epsilon$ faster than
$\textit{O}(\epsilon^N)$ for any $N$. (Notice that
$\frac{\varepsilon}{{\Delta}^2}\propto \epsilon$, i.e., the standard
adiabatic condition guarantees error $\textit{O}(\epsilon^2)$.)

In our case, the smoothness condition translates directly to the
functions $f(t)$ and $g(t)$. For any smooth $f(t)$ and $g(t)$ we can
further ensure that the condition at the end points is satisfied by
the reparameterization $f(t)\rightarrow f(y(t))$, $g(t)\rightarrow
g(y(t))$ where $y(t)$ is a smooth function of $t$ which satisfies
$y(0)=0$, $y(T)=T$, and has vanishing derivatives at $t=0$ and
$t=T$. Then by slowing down the change of the Hamiltonian by a
constant factor $\epsilon$, which amounts to an increase of the
total time $T$ by a factor $1/\epsilon$, we can decrease the error
super-polynomially in $\epsilon$. We will use this result to obtain
a low-error interpolation in Sec.~\ref{e4:FTofthescheme} where we
estimate the time needed to implement a holonomic gate with a given
precision.

\subsubsection{The CNOT gate}\label{e4:sectionCNOT}

The stabilizer (or gauge group) on multiple blocks of the code is a
direct product of the stabilizers (or gauge groups) of the
individual blocks. Therefore, from Eq.~\eqref{e4:stabelementgen} it
follows that one can always find an element of the initial
stabilizer or gauge group on multiple blocks that has any desired
combination of factors $\sigma^i$, $i=0,1,2,3$ on the first qubits
in these blocks. It can be verified that applying transversal\index{transversal!gate!Clifford}
\index{Clifford!operation} Clifford operations on the blocks does not change this property.
Therefore, we can find an element of the stabilizer or the gauge
group that has the form \eqref{e4:stabelement} where the factor $Z$
acts on the target qubit and $\widetilde{G}$ acts trivially on the
control qubit. We now explain how to implement the CNOT gate
adiabatically starting from such a Hamiltonian.

Notice that a Hamiltonian of the form
\begin{equation}
\widehat{\widehat{H}}(t)=|0\rangle\langle 0|^c \otimes
{H}_{0}(t)\otimes\widetilde{G}+|1\rangle\langle 1|^c\otimes
{H}_{1}(t)\otimes\widetilde{G},\label{e4:hathatH}
\end{equation}
where the superscript $c$ denotes the control qubit, gives rise to
the unitary transformation
\begin{equation}
\widehat{\widehat{U}}(t)=|0\rangle\langle 0|^c \otimes
\widehat{U}_0(t)+|1\rangle\langle 1|^c\otimes
\widehat{U}_1(t),\label{e4:doublehatUU}
\end{equation}
where
\begin{equation}
\widehat{U}_{0,1}(t)=\mathcal{T}\textrm{exp}(-i\overset{t}{\underset{0}{\int}}d\tau
H_{0,1}(\tau)\otimes\widetilde{G}).
\end{equation}
If $H_0(t)$ and $H_1(t)$ have the same non-degenerate instantaneous
spectra, and $\textrm{Tr}\{H_{0,1}(t)\}=0$, then it follows that in
the adiabatic limit each of the eigenspaces of
$\widehat{\widehat{H}}(t)$ will undergo the transformation
\begin{equation}
\widehat{\widehat{U}}_g(t)=|0\rangle\langle 0|^c \otimes
V_0(t)\otimes \widetilde{I}+|1\rangle\langle 1|^c\otimes
V_1(t)\otimes \widetilde{I},\label{e4:CNOTtransformation}
\end{equation}
where $V_{0,1}(t)\otimes \widetilde{I}$ are the effective
transformations generated by ${H}_{0,1}(t)\otimes\widetilde{G}$
according to Proposition~\ref{e4:Proposition1sub}.

Our goal is to find $H_0(t)$ and $H_1(t)$, $H_0(0)=H_1(0)=Z$, such
that at the end of the transformation the unitary
\eqref{e4:CNOTtransformation} will be equal to the CNOT gate. In other
words, we want $V_0(2T)=I$ and $V_1(2T)=X$ (we choose the total time
of evolution to be $2T$ for convenience).

We already saw how to generate adiabatically the $X$ gate up to a
phase---Eq.~\eqref{e4:Xgate}. We can use the same Hamiltonian in place
of $H_1(t)$:
\begin{equation}
H_1(t)=\begin{cases} H^{\pi/2+}(t), \hspace{0.2cm} 0\leq t \leq T\\
H^{\pi/2-}(2T-t), \hspace{0.2cm} T\leq t \leq 2T.
\end{cases}\label{e4:H1t}
\end{equation}

Now we want to find a Hamiltonian $H_0(t)$ with the same spectrum as
$H_1(t)$, which gives rise to a trivial geometric transformation,
$V_0(2T)=I$ (possibly up to a phase, which can be undone later).
Since all Hamiltonians of the type $H^{\theta\pm}(t)$ have the same
instantaneous spectrum (for fixed $f(t)$ and $g(t)$), we can simply
choose
\begin{equation}
H_0(t)=\begin{cases} H^{\pi/2+}(t), \hspace{0.2cm} 0\leq t \leq T\\
H^{\pi/2+}(2T-t), \hspace{0.2cm} T\leq t \leq 2T,
\end{cases}\label{e4:H0t}
\end{equation}
which corresponds to applying a given transformation from $t=0$ to
$t=T$ and then undoing it (running it backwards) from $t=T$ to
$t=2T$. This results exactly in $V_0(2T)=I$.

Since, as we saw in Sec.~\ref{e4:sectionSingleQubit}, the Hamiltonian
$H_1(t)\otimes \widetilde{G}$ gives rise to the transformation
$iX\otimes\widetilde I$ in each of its eigenspaces, the above choice
for the Hamiltonians \eqref{e4:H0t} and \eqref{e4:H1t} in
Eq.~\eqref{e4:hathatH} will result in the transformation
\begin{equation}
|0\rangle\langle 0|^c \otimes
I\otimes\widetilde{I}+i|1\rangle\langle 1|^c\otimes
{X}\otimes\widetilde{I},
\end{equation}
which is the desired CNOT gate up to a Phase gate\index{phase!gate} on the control
qubit. We can correct the phase by applying the inverse of the Phase
gate to the control qubit, either before or after the described
transformation.

Notice that from $t=0$ to $t=T$ the Hamiltonians \eqref{e4:H0t} and
\eqref{e4:H1t} are identical, i.e., during this period the Hamiltonian
\eqref{e4:hathatH} has the form
\begin{equation}
I^c\otimes H^{\pi/2+}(t)\otimes\widetilde{G},
\end{equation}
so we are simply applying the single-qubit operation $V^{\pi/2+}$ to
the target qubit according to the method for single-qubit gates
described before. It is easy to verify that during the second
period, from $t=T$ to $t=2T$, the Hamiltonian \eqref{e4:hathatH}
realizes the interpolation
\begin{equation}
-I^c\otimes Y\otimes \widetilde{G} \rightarrow -Z^c\otimes Z\otimes
\widetilde{G},\label{e4:CNOTint}
\end{equation}
which is understood as in Eq.~\eqref{e4:interpolation}.

To summarize, the CNOT gate can be implemented by first applying the
inverse of the Phase gate ($P^{\dagger}$) on the control qubit, as
well as the transformation $V^{\pi/2+}$ on the target qubit,
followed by the transformation \eqref{e4:CNOTint}. Due to the form
\eqref{e4:hathatH} of $\widehat{\widehat{H}}(t)$, the extent to which
the adiabatic approximation is satisfied during this transformation
depends only on the adiabatic properties of the single-qubit
Hamiltonians $H^{\pi/2\pm}(t)$ which we discussed in the previous
subsection.

\subsection{Encoded operations outside of the Clifford group}

For universal fault-tolerant computation\index{universal!fault-tolerant quantum computation} we also need at least one
encoded gate outside of the Clifford group.\index{group!Clifford}\index{Clifford!group}\index{group!Clifford}
 The fault-tolerant
implementation of such gates is based on the preparation of a
special encoded state
\cite{Shor:1996:56,Knill:1998:365,Gottesman:9705052, Boykin:1999:486,
Zhou:2000:052316} which involves a measurement of an encoded
operator in the Clifford group. For example, the $\pi/8$ gate\index{quantum gate!$\pi/8$}
requires the preparation of the encoded state
$\frac{|0\rangle+\textrm{exp}(i\pi/4)|1\rangle}{\sqrt{2}}$, which
can be realized by measuring the encoded operator $e^{-i\pi/4}PX$
\cite{Boykin:1999:486}. Equivalently, the state can be obtained by
applying the encoded operation $WP^{\dagger}$ on the encoded state
$\frac{\cos(\pi/8)|0\rangle+\sin(\pi/8)|1\rangle}{\sqrt{2}}$ which
can be prepared by measuring the encoded Hadamard gate
\cite{Knill:1998:365}. The Toffoli gate requires the preparation of
the three-qubit encoded state
$\frac{|000\rangle+|010\rangle+|100\rangle+|111\rangle}{2}$ and
involves a similar procedure \cite{Zhou:2000:052316}. In all these
instances, the measurement of the encoded Clifford operator is
realized by applying transversally the operator conditioned on the
qubits in a ``cat'' state.

We now describe a general method that can be used to implement
adiabatically any conditional transversal\index{transversal!gate!adiabatic}\index{Clifford!operation} Clifford operation with
conditioning on the ``cat'' state. Let $O$ be a Clifford gate acting
on the first qubits from some set of blocks. As we discussed in the
previous section, under this unitary the stabilizer and the gauge
group transform in such a way that we can always find an element
with an arbitrary combination of Pauli matrices on the first qubits.
If we write this element in the form
\begin{equation}
\widehat{G}=G_1\otimes G_{2,...,n},
\end{equation}
where $G_1$ is a tensor product of Pauli matrices acting on the
first qubits from the blocks, and $G_{2,...,n}$ is an operator on
the rest of the qubits, then applying $O$ conditioned on the first
qubit in a ``cat'' state transforms this stabilizer or gauge-group
element as follows:
\begin{gather}
I^c\otimes G_1\otimes G_{2,...,n}=|0\rangle\langle 0|^c\otimes
G_1\otimes G_{2,...,n}+|1\rangle\langle 1|^c\otimes G_1\otimes
G_{2,...,n} \notag\\\rightarrow |0\rangle\langle 0|^c\otimes G_1\otimes
G_{2,...,n}+|1\rangle\langle 1|^c\otimes OG_1O^{\dagger}\otimes
G_{2,...,n},
\end{gather}
where the superscript $c$ denotes the control qubit from the
``cat'' state. We can implement this operation by choosing the
factor $G_1$ to be the same as the one we would use if we wanted to
implement the operation $O$ according to the previously described
procedure. Then we can apply the following Hamiltonian:
\begin{equation}
\widehat{\widehat{H}}_{C(O)}(t)=-|0\rangle\langle 0|^c\otimes
G_1\otimes G_{2,...,n}-\alpha(t)|1\rangle\langle 1|^c\otimes
H_O(t)\otimes G_{2,...,n},\label{e4:HamA}
\end{equation}
where $H_O(t)\otimes G_{2,...,n}$ is the Hamiltonian that we would
use for the implementation of the operation $O$ and $\alpha(t)$ is a
real parameter chosen such that at every moment the operator
$\alpha(t)|1\rangle\langle 1|^c\otimes H_O(t)\otimes G_{2,...,n}$
has the same instantaneous spectrum as the operator
$|0\rangle\langle 0|^c\otimes G_1\otimes G_{2,...,n}$. This
guarantees that the overall Hamiltonian is degenerate and the
transformation on each of its eigenspaces is
\begin{equation}
\widehat{\widehat{U}}_g(t)=|0\rangle\langle 0|^c\otimes I_1\otimes
I_{2,...,n}+|1\rangle\langle 1|^c\otimes U_O(t)\otimes
I_{2,...,n},
\end{equation}
where $U_O(t)$ is the transformation on the first qubits generated
by $H_O(t)\otimes G_{2,...,n}$. Since we presented the constructions
of the basic\index{Clifford!operation} Clifford operations up to an overall phase, the
operation $U_O(t)$ may differ from the desired operation by a phase.
This phase can be corrected by applying a suitable gate on the
control qubit from the ``cat'' state (we explain how this can be done
in the next section). We remark that a Hamiltonian of the type
\eqref{e4:HamA} requires fine tuning of the parameter $\alpha(t)$ and
generally can be complicated. In Sec.~\ref{e4:BScode3local} we will
show that depending on the code one can find more natural
implementations of these operations.

If we want to apply a second conditional\index{Clifford!operation} Clifford operation $Q$ on
the first qubits in the blocks, we can do this as follows. Imagine
that if we had to apply the operation $Q$ following the operation
$O$, we would use the Hamiltonian $\widehat{H}_Q(t)=-H_Q(t)\otimes
G'_{2,...,n}$, where $\widehat{H}_Q(0)=OG'_1O^{\dagger}\otimes
G'_{2,...,n}$ is a suitable element of the stabilizer or the gauge
group after the application of $O$. Before the application of $O$,
that element would have had the form $G'_1\otimes G'_{2,...,n}$.
Under the application of a conditional $O$, the element $G'_1\otimes
G'_{2,...,n}$ transforms to $|0\rangle\langle 0|^c\otimes
G'_1\otimes G'_{2,...,n}+|1\rangle\langle 1|^c\otimes
OG'_1O^{\dagger}\otimes G'_{2,...,n}$ which can be used (with a
minus sign) as a starting Hamiltonian for a subsequent operation. In
particular, we can implement the conditional $Q$ following the
conditional $O$ using the Hamiltonian
\begin{equation}
\widehat{\widehat{H}}_{C(Q)}(t)=-|0\rangle\langle 0|^c\otimes
G'_1\otimes G'_{2,...,n}-\beta(t)|1\rangle\langle 1|^c\otimes
H_Q(t)\otimes G'_{2,...,n}, \label{e4:HamB}
\end{equation}
where the factor $\beta(t)$ guarantees that there is no splitting of
the energy  levels. Subsequent operations can be applied analogously.
Using this general method, we can implement a unitary whose
geometric part is equal to any transversal\index{Clifford!operation} Clifford operation
conditioned on the ``cat" state.

\subsection{Preparing and using the ``cat'' state}\label{e4:sectionCAT}

In addition to\index{transversal!operation} transversal operations, a complete fault-tolerant
scheme requires the ability to prepare, verify and use a special
ancillary state such as the ``cat'' state
$(|00...0\rangle+|11...1\rangle)/\sqrt{2}$. This can also be done in
the spirit of the described holonomic scheme. Since the ``cat'' state
is known and its construction is non-fault-tolerant, we can prepare
it by simply treating each initially prepared qubit as a simple code
(with $\widetilde{G}$ in Eq.~\eqref{e4:stabelement} being trivial), and
updating the stabilizer of the code via the applied transformation
as the operation progresses. The stabilizer of the prepared ``cat''
state is generated by $Z_iZ_j$, $i<j$. Transversal unitary
operations between the ``cat'' state and other codewords are applied
as described in the previous sections.

We also have to be able to measure the parity of the state, which
requires the ability to apply successive CNOT operations from two
different qubits in the ``cat'' state to the same ancillary qubit
initially prepared in the state $|0\rangle$. We
can regard a qubit in state $|0\rangle$ as a simple code with stabilizer $%
\langle Z \rangle $, and we can apply the first CNOT as described
before. Even though after this operation the state of the target
qubit is unknown, the second CNOT gate can be applied via the same
interaction, since the transformation in each eigenspace of the
Hamiltonian is the same and at the end when we measure the qubit we
project onto one of the eigenspaces.

\subsection{Fault tolerance of the scheme}\label{e4:FTofthescheme}

We saw how we can transform the code subsystem by any transversal
operation adiabatically, which allows us, by a sequence of such
operations, to generate a\index{holonomy} holonomy in the subsystem which is
equivalent to any encoded gate. This was based on the assumption
that the state has not undergone an error. But what if an error
occurs on one of the qubits?

At any moment, we can distinguish two types of errors---those that
result in transitions between the ground and the excited spaces of
the current Hamiltonian, and those that result in transformations
inside the eigenspaces. Due to the discretization of errors in
quantum error correction, it suffices to prove correctability for
each type separately. The key property of the construction we
described is that each of the eigenspaces undergoes the same
transversal transformation. Because of this, if we are applying a
unitary on the first qubit, an error on that qubit will remain
localized regardless of whether it causes an excitation or not. If
the error occurs on one of the other qubits, at the end of the
transformation the result would be the desired single-qubit unitary
gate plus the error on the other qubit, which is correctable.

We see that even though the Hamiltonian couples qubits in the same
block, single-qubit errors do not propagate. This is because the
coupling between the qubits amounts to a change in the relative
phase between the ground and excited spaces, but the latter is
irrelevant since either it is equivalent to a gauge transformation,
or when we apply a correcting operation we project on one of the
eigenspaces. In the case of the CNOT gate, an error can propagate
between the control and the target qubits, but it never results in
two errors within the same codeword.

\subsection{FTHQC with the Bacon-Shor
code\index{code!Bacon-Shor} using 3-local Hamiltonians}\label{e4:BScode3local}

The weight of the Hamiltonians needed for the scheme we described
depends on the weight of the stabilizer or gauge-group elements.
Remarkably, certain codes possess stabilizer or gauge-group elements
of low weight covering all qubits in the code, which allows us to
perform holonomic computation using low-weight Hamiltonians. Here we
will consider as an example a subsystem generalization of the
9-qubit Shor code \cite{Shor:1995:R2493}---the Bacon-Shor code
\cite{Bacon:2006:012340, Bacon:0610088}---which has particularly
favorable properties for fault-tolerant computation
\cite{Aliferis:2007:220502, Aliferis:0703230}. In the 9-qubit
Bacon-Shor code, the gauge group is generated by the weight-two
operators $Z_{k,j}Z_{k,j+1}$ and $X_{j,k}X_{j+1,k}$, where the
subscripts label the qubits by row and column when they are arranged
in a $3\times 3$ square lattice. Since the Bacon-Shor code is a CSS
code, the CNOT gate has a direct transversal\index{transversal!gate!CNOT} implementation. We now
show that the CNOT gate can be realized using at most weight-three
Hamiltonians.

If we want to apply a CNOT gate between two qubits each of which is,
say, in the first row and column of its block, we can use as a
starting Hamiltonian $-Z^{t}_{1,1}\otimes Z^t_{1,2}$, where the
superscript $t$ signifies that these are operators in the target
block. We can then apply the CNOT gate as described in the previous
section. After the operation, however, this gauge-group element will
transform to $-Z^{t}_{1,1}\otimes Z^c_{1,1}\otimes Z^t_{1,2}$. If we
now want to implement a CNOT gate between the qubits with index
$\{1,2\}$ using as a starting Hamiltonian the operator
$-Z^{t}_{1,1}\otimes Z^c_{1,1}\otimes Z^t_{1,2}$ according to the
same procedure, we will have to use a four-qubit Hamiltonian. Of
course, at this point we can use the starting Hamiltonian
$-Z^t_{1,2}\otimes Z^t_{1,3}$, but if we had also applied a CNOT
between the qubits labelled $\{1,3\}$, this operator would not be
available---it would have transformed to $-Z^t_{1,2}\otimes
Z^t_{1,3}\otimes Z^c_{1,3}$.

What we can do instead, is to use as a starting Hamiltonian the
operator $ -Z^{t}_{1,1}\otimes Z^t_{1,2}\otimes Z^c_{1,2}$ which is
obtained from the gauge-group element $ Z^{t}_{1,1}\otimes
Z^c_{1,1}\otimes Z^t_{1,2}\otimes Z^c_{1,2}$ after the application
of the CNOT between the qubits with index $\{1,1\}$. Since the CNOT
gate is its own inverse, we can regard the factor $Z^{t}_{1,1}$ as
$\widetilde{G}$ in Eq.~\eqref{e4:CNOTint} and use this starting
Hamiltonian to apply the procedure backwards. Thus we can implement
any transversal CNOT gate using at most weight-three Hamiltonians.

Since the encoded $X$, $Y$ and $Z$ operations have a bitwise
implementation, we can always apply them according to the described
procedure using Hamiltonians of weight 2. For the Bacon-Shor code,
the encoded Hadamard gate can be applied via bitwise Hadamard
transformations followed by a rotation of the grid by a $90$ degree
angle \cite{Aliferis:2007:220502}. The encoded $P$ gate can be
implemented by using the encoded CNOT and an ancilla.

The preparation and measurement of the ``cat'' state can also be
done using Hamiltonians of weight 2. To prepare the ``cat'' state,
we prepare first all qubits in the state $(|0\rangle +
|1\rangle)/\sqrt{2}$, which can be done by measuring each of them in
the $\{|0\rangle,|1\rangle\}$ basis (this ability is assumed for any
type of computation) and applying the transformation $-Z\rightarrow
-X$ or $Z\rightarrow -X$ depending on the outcome. To complete the
preparation of the ``cat'' state, apply a two-qubit transformation
between the first qubit and each of the other qubits ($j>1$) via the
transformation
\begin{equation}
-I_1\otimes X_j\rightarrow -Z_1\otimes Z_j .
\end{equation}
Single-qubit transformations on qubits from the ``cat'' state can
be applied according to the method described in the previous
section using at most weight-two Hamiltonians.

To measure the parity of the state, we need to apply successively
CNOT operations from two different qubits in the ``cat'' state to the
same ancillary qubit initially prepared in the state $|0\rangle$.
This can also be done according to the method described in
Sec.~\ref{e4:sectionCAT} and requires Hamiltonians of weight 2.

For universal computation with the Bacon-Shor code, we also need to 
be able to apply one encoded transformation outside of the Clifford
group. As we mentioned earlier, in order to implement the Toffoli
gate or the $\pi/8$ gate, it is sufficient to be able to implement a
CNOT gate conditioned on a ``cat'' state. For the Bacon-Shor code,
the CNOT gate has a transversal implementation, so the conditioned
CNOT gate can be realized by a series of transversal Toffoli
operations between the ``cat'' state and the two encoded states. We
will now show that this gate can be implemented using at most
three-qubit Hamiltonians.

Ref.~\cite{Nielsen:2000:CambridgeUniversityPress} provides a circuit for
implementing the Toffoli gate as a sequence of one- and two-qubit
gates. We will use the same circuit, except that we flip the control
and target qubits in every CNOT gate using the identity
\begin{equation}
(W_{1}W_{2}) CX_{1,2}(W_{1}W_{2})= CX_{2,1},
\end{equation}
where $W_{i}$ denotes a Hadamard gate on the qubit labelled by $i$
and $CX_{i,j}$ denotes a CNOT gate between qubits $i$ and $j$ with
$i$ being the control and $j$ being the target. Let
$\textrm{Toffoli}_{i,j,k}$ denote the Toffoli gate on qubits $i$,
$j$ and $k$ with $i$ and $j$ being the two control qubits and $k$
being the target qubit, and let $P_{i}$ and $T_i$ denote the Phase
and $\pi/8$ gates on qubit $i$, respectively. Then the Toffoli gate
on three qubits (the first one of which we will assume to belong to
the ``cat'' state), can be written as:
\begin{gather}
\textrm{Toffoli}_{1,2,3}=W_{2}CX_{3,2}W_{3}T_{3}^{\dagger}W_{3}W_{1}CX_{3,1}W_{3}T_{3}W_{3}CX_{3,2}W_{3}T_{3}^{\dagger}\notag\\\times
W_{3} CX_{3,1}W_{3}T_{3}W_{3} W_{2}T_{2}^{\dagger}W_{2} CX_{2,1}
W_{2}T_{2}^{\dagger}W_{2}
CX_{2,1}W_{2}P_{2}W_{1}T_{1}.\label{e4:Toffoli}
\end{gather}
To show that each of the above gates can be implemented
adiabatically using Hamiltonians of weight at most 3, we will need
an implementation of the CNOT gate which is suitable for the case
when we have a stabilizer or gauge-group element of the form
\begin{equation}
\widehat{G}=X\otimes \widetilde{G},\label{e4:Xgenerator}
\end{equation}
where the factor $X$ acts on the target qubit and $\widetilde{G}$
acts trivially on the control qubit. By a similar argument to the
one in Sec.~\ref{e4:sectionCNOT}, one can verify that in this case the
CNOT gate can be implemented as follows: apply the operation
$P^{\dagger}$ on the control qubit (we describe how to do this for
our particular case below) together with the transformation
\begin{equation}
-X\otimes\widetilde{G}\rightarrow - Z\otimes
\widetilde{G}\rightarrow X\otimes\widetilde{G}\label{e4:CNOT1}
\end{equation}
on the target qubit, followed by the transformation
\begin{equation}
I^c\otimes X\otimes \widetilde{G}\rightarrow -(|0\rangle\langle
0|^c\otimes Z+|1\rangle\langle 1|^c\otimes Y)\otimes
\widetilde{G}\rightarrow-I^c\otimes X\otimes
\widetilde{G}.\label{e4:CNOT2}
\end{equation}

Since the second and the third qubits belong to blocks encoded with
the Bacon-Shor code, there are weight-two elements of the initial
gauge group of the form $Z\otimes Z$ covering all qubits. The
stabilizer generators on the ``cat'' state are also of this type.
Following the transformation of these operators according to the
sequence of operations \eqref{e4:Toffoli}, one can see that before
every CNOT gate in this sequence, there is an element of the form
\eqref{e4:Xgenerator} with $\widetilde{G}=Z$ that can be used to
implement the CNOT gate as described, provided that we can implement
the gate $P^{\dagger}$ on the control qubit. We also point out that
all single-qubit operations on qubit $1$ in this sequence can be
implemented according to the procedure describes in
Sec.~\ref{e4:sectionSingleQubit}, since at every step we have a
weight-two stabilizer element on that qubit with a suitable form.
Therefore, all we need to show is how to implement the necessary
single-qubit operations\index{single-qubit operation} on qubits $2$ and $3$. Due to the
complicated transformation of the gauge-group elements during the
sequence of operations \eqref{e4:Toffoli}, we will introduce a method
of applying a single-qubit operation with a starting Hamiltonian
that acts trivially on the qubit. For implementing single-qubit
operations on qubits $2$ and $3$ we will use as a starting
Hamiltonian the operator
\begin{equation}
\widehat{\widehat{H}}(0)=-I_{i}\otimes X_1\otimes \widetilde{Z},
\hspace{0.4 cm} i=2,3\label{e4:specialcase}
\end{equation}
where the first factor ($I_i$) acts on the qubit on which we want
to apply the operation ($2$ or $3$), and $X_1\otimes
\widetilde{Z}$ is the transformed (after the Hadamard gate $R_1$)
stabilizer element of the ``cat'' state that acts non-trivially on
qubit $1$ (the factor $\widetilde{Z}$ acts on some other qubit in
the ``cat'' state).

To implement a single-qubit gate on qubit $3$ for example, we
first apply the interpolation
\begin{equation}
-I_{3}\otimes X_1\otimes \widetilde{Z}\rightarrow -Z_{3}\otimes
Z_1\otimes \widetilde{Z}.\label{e4:nexttolast}
\end{equation}
This results in a two-qubit transformation $U_{1,3}$ on qubits $1$
and $3$ in each eigenspace. We do not have to calculate this
transformation exactly since we will undo it later, but the fact
that each eigenspace undergoes the same two-qubit transformation can
be verified similarly to the CNOT gate we described in
Sec.~\ref{e4:sectionCNOT}.

At this point, the Hamiltonian is of the form \eqref{e4:H(0)} with
respect to qubit 3, and we can apply any single-qubit unitary gate
$V_3$ according to the method described in
Sec.~\ref{e4:sectionSingleQubit}. This transforms the Hamiltonian to
$-V_3Z_{3}V_3^{\dagger}\otimes Z_1\otimes \widetilde{Z}$. We can now
``undo'' the transformation $U_{1,3}$ by the interpolation
\begin{equation}
-V_3Z_{3}V_{3}^{\dagger}\otimes Z_1\otimes
\widetilde{Z}\rightarrow-I_{3}\otimes X_1\otimes
\widetilde{Z}.\label{e4:lastHamiltonian}
\end{equation}
The latter transformation is the inverse of Eq.~\eqref{e4:nexttolast}
up to the single-qubit unitary transformation $V_3$, i.e., it
results in the transformation $V_3U_{1,3}^{\dagger}V_{3}^{\dagger}$.
Thus, the net result is
\begin{equation}
V_3U_{1,3}^{\dagger}V_3^{\dagger}V_3U_{1,3}=V_3,
\end{equation}
which is the desired single-qubit unitary transformation on qubit
$3$. Note that during this transformation, a single-qubit error can
propagate between qubits $1$ and $3$, but this is not a problem
since we are implementing a transversal Toffoli\index{transversal!gate!Toffoli} operation and such
an error would not result in more that one error per block of the
code.

We see that with the Bacon-Shor code\index{code!Bacon-Shor} the above scheme for
fault-tolerant HQC can be implemented with at most 3-local
Hamiltonians. This is optimal for this approach, because there are
no non-trivial codes with stabilizer or gauge-group elements of
weight smaller than 2 covering all qubits.

\subsection{Effects on the accuracy threshold}

Since the method we described conforms completely to a given
fault-tolerant scheme, it would not affect the error threshold\index{error!threshold} per
qubit per operation for that scheme. However, the allowed errors per
qubit per operation include both errors due to imperfectly applied
transformations and errors due to interaction with the environment.
It turns out that certain features of the method we described have an
effect on the allowed distribution of these errors within the
accuracy threshold.

First, observe that when applying the Hamiltonian \eqref{e4:Ham1}, we
cannot at the same time apply
operations on the other qubits on which the factor $%
\widetilde{G}$ acts non-trivially. Thus, some operations at the
lowest level of concatenation\index{concatenation!for fault tolerance of holonomic QC} that would otherwise be implemented
simultaneously might have to be implemented serially. The effect of
this is equivalent to slowing down the circuit by a constant factor.
(Note that we could also vary the factor $\widetilde{G}$
simultaneously with $H(t)$, but in order to obtain the same
precision as that we would achieve by a serial implementation, we
would have to slow down the change of the Hamiltonian by the same
factor.) The slowdown factor resulting from this loss of parallelism
is usually small since this problem occurs only at the lowest level
of concatenation. It can be verified that for the Bacon-Shor code,
we can apply operations on up to 6 out of the 9 qubits in a block
simultaneously. For example, when applying encoded single-qubit
operations, we can address simultaneously any two qubits in a row or
column by taking $\widetilde{G}$ to be a single-qubit operator $Z$
or $X$ on the third qubit in the same row or column. The
Hamiltonians used to apply operations on the two qubits commute with
each other at all times and do not interfere. A similar phenomenon
holds for the implementation of the encoded CNOT gate or the
operations involving the ``cat'' state. Thus, for the Bacon-Shor
code we have a slowdown due to parallelism by a factor of $1.5$.

A more significant slowdown results from the fact that the evolution
is adiabatic. In order to obtain a rough estimate of the slowdown
due specifically to the adiabatic requirement, we will compare the
time $T_{h}$ needed for the implementation of a gate according to
our method with precision $1-\delta$ to the time $T_{d}$ needed for
a dynamical realization of the same gate with the same strength of
the Hamiltonian. We will consider a realization of the $X$ gate via
the unitary interpolation
\begin{equation}
\widehat{H}(t)=-V_{X}(\tau(t))ZV_{X}^{\dagger}(\tau(t))\otimes
\widetilde{G}, \hspace{0.2cm} V_{X}(\tau(t))=\textrm{exp}
\left(i\tau(t)\frac{\pi }{2T_{h}}X\right),\label{e4:Hamest}
\end{equation}
where $\tau(0)=0$, $\tau(T_h)=T_h$. The energy gap of this
Hamiltonian is constant. The optimal dynamical implementation of the
same gate is via the Hamiltonian $-X$ for time
$T_{d}=\frac{\pi}{2}$.

As we argued in Sec.~\ref{e4:sectionAdicond}, the accuracy with which
the adiabatic approximation holds for the Hamiltonian \eqref{e4:Hamest}
is the same as that for the Hamiltonian
\begin{equation}
H(t)=V_{X}(\tau(t))ZV_{X}^{\dagger }(\tau(t)).\label{e4:Hamest2}
\end{equation}
We now present estimates for two different choices of the function
$\tau(t)$. The first one is
\begin{equation}
\tau(t)=t.
\end{equation}
In this case the Schr\"{o}dinger equation can be easily solved in
the instantaneous eigenbasis of the Hamiltonian \eqref{e4:Hamest2}.
For the probability that the initial ground state remains a ground
state at the end of the evolution, we obtain
\begin{equation}
p=\frac{1}{1+\varepsilon ^{2}}+\frac{%
\varepsilon ^{2}}{1+\varepsilon ^{2}}\cos ^{2}(\frac{\pi }{4\varepsilon }%
\sqrt{1+\varepsilon ^{2}})=1-\delta,
\end{equation}
where
\begin{equation}
\varepsilon =\frac{T_{d}}{T_{h}}.
\end{equation}
Expanding in powers of $\varepsilon $ and averaging the square of
the cosine whose period is much smaller than $T_{h}$, we obtain
the condition
\begin{equation}
\varepsilon ^{2}\leq 2\delta.
\end{equation}
Assuming, for example, that $\delta \approx 10^{-4}$ (approximately
the threshold for the 9-qubit Bacon-Shor code\index{code!Bacon-Shor}
\cite{Aliferis:2007:220502}), we obtain that the time of evolution
for the adiabatic case must be about 70 times longer than that in
the dynamical case.

It is known, however, that if $H(t)$ is smooth and its derivatives
vanish at $t=0$ and $t=T_h$, the adiabatic error decreases
super-polynomially with $T_h$ \cite{Hagedorn:2002:235}. To achieve
this, we will choose
\begin{equation}
\tau(t)= \frac{1}{a}\int_0^t dt' e^{-1/\sin(\pi
t'/T_h)},\hspace{0.2cm} a=\int_0^{T_h} dt' e^{-1/\sin(\pi
t'/T_h)}.
\end{equation}
For this interpolation, by a numerical solution we obtain that
when $T_h/T_d\approx 17$ the error is already of the order of
$10^{-6}$, which is well below the threshold values obtained for
the Bacon-Shor codes \cite{Aliferis:2007:220502}. This is a remarkable improvement
in comparison to the previous interpolation which shows that the
smoothness of the Hamiltonian plays an important role in the
performance of the scheme.

An additional slowdown in comparison to a perfect dynamical scheme
may result from the fact that the constructions for some of the
standard gates we presented involve long sequences of loops. With
more efficient parameter paths, however, it should be possible to
reduce this slowdown to minimum.

In comparison to a dynamical implementation, the allowed rate of
environmental noise for the holonomic case would decrease by a
factor similar to the slowdown factor. In practice, however,
dynamical gates are not perfect and the holonomic approach may be
advantageous if it allows for a better precision.

We finally point out that an error in the factor $H(t)$ in the
Hamiltonian \eqref{e4:Ham1} would result in an error on the first qubit
according to Eq.~\eqref{e4:finalU}. Such an error clearly has to be
below the accuracy threshold. More dangerous errors, however, are
also possible. For example, if the degeneracy of the Hamiltonian is
broken, this can result in an unwanted dynamical transformation
affecting all qubits on which the Hamiltonian acts non-trivially.
Such multi-qubit errors have to be of higher order in the threshold,
which imposes more severe restrictions.

\section{Conclusion and outlook}

In this chapter we saw that HQC can be made fault tolerant by
combining it with the techniques for fault-tolerant quantum error
correction on stabilizer codes. This means that HQC is, at least in
principle, a scalable method of computation. However, further
research is needed in order to bring the presented ideas closer to
experimental realization.

We presented a scheme which uses Hamiltonians that are elements of
the stabilizer or the gauge group of the code. We saw that with the
Bacon-Shor code, this scheme can be implemented with 2- and 3-qubit
Hamiltonians. Since the scheme conforms completely to a given
dynamical fault-tolerant scheme and does not require the use of
extra qubits, it has the same error threshold as the dynamical
scheme on which it is based. However, due to the fact that adiabatic
gates are slower than dynamical gates and that at the lowest
concatenation level this scheme requires certain gates to be
implemented serially or more slowly, the allowed error threshold for
environmental noise is lower in comparison to a dynamical scheme.
The factor by which this threshold decreases depends on how smooth
the adiabatic interpolations are, but it seems to be at most $\sim
10^2$. Therefore, if the robustness provided by the geometric nature
of the gates is sufficiently higher than that achievable by
dynamical means, this approach could be advantageous in comparison
to dynamical schemes. The main challenge in the implementation of
this is approach, however, is that it requires the engineering of
3-local Hamiltonians. An alternative scheme for fault-tolerant HQC that uses Hamiltonians independent of the code at the expense of additional qubits \cite{Oreshkov:2009:090502} has been proven to allow reducing the\index{locality} locality
 of the Hamiltonian with perturbative gadget techniques, showing that 2-local Hamiltonians are universal for\index{universality!of 2-local Hamiltonians}
fault-tolerant HQC. The disadvantage of using the gadgets is that
they decrease the gap of the Hamiltonian by a very large factor,
which requires a significant slowdown of the computation and
decreases the allowed rate of environment noise.

Applying the strategies we have described
to actual physical systems will
undoubtedly require modifications in accordance with the available
interactions in those systems. A possible way of avoiding the use of
multi-local Hamiltonians without the use of perturbative gadgets
could be to use higher-dimensional systems (e.g., qutrits)\index{qutrit} between
which two-local interactions are naturally available. It may be
possible to encode qubits in subspaces or subsystems of these
higher-dimensional systems and use fault-tolerant techniques
designed for stabilizer codes based on\index{qudit} qudits
\cite{Gottesman:1999:1749}. Given that simple quantum
error-correcting codes and two-qubit geometric transformations have been realized using NMR\index{NMR (nuclear magnetic resonance)} \cite{Cory:1998:2152, Jones:2000:869} and ion-trap \cite%
{Chiaverini:2004:602, Leibfried:2003:412} techniques, these systems seem particularly suitable
for hybrid HQC-QEC implementations.

Finally, it is interesting to point out that the adiabatic regime in
which the holonomic schemes operate is consistent with the Markovian
model of decoherence. In Ref.~\cite{Alicki:2006:052311} it was
argued that the standard dynamical paradigm of fault tolerance is
based on assumptions that are in conflict with the rigorous
derivation of the Markovian limit. Although the threshold theorem
\index{theorem!threshold}
has been extended to non-Markovian models \cite{Terhal:2005:012336,
Aliferis:2006:97, Aharonov:2006:050504}, the Markovian assumption is
an accurate approximation for a wide range of physical scenarios
\cite{Carmichael:1993:Springer} and allows for a much simpler
description of the evolution in comparison to non-Markovian models
(see Chapter 8). In Ref.~\cite{Alicki:2006:052311} it was
shown that the weak-coupling-limit derivation of the Markovian
approximation is consistent with computational methods that employ
slow transformations, such as adiabatic quantum computation
\cite{Farhi:0001106} or HQC. A theory of fault-tolerance for the
adiabatic model of computation at present is not known, although
steps in this direction have been undertaken (see
Refs.~\cite{Jordan:2005:052322, Lidar:2008:160506}). The hybrid HQC-QEC schemes presented here provide
solutions for the case of HQC. However, we point out that it is an
open problem whether the Markovian approximation makes sense for a
fixed value of the adiabatic slowness parameter when the circuit
increases in size. Giving a definitive answer to this question
requires a rigorous analysis of the accumulation of non-Markovian
errors due to deviation from perfect adiabaticity.

The techniques described in this chapter may prove useful in other
areas as well. It is possible that some combination of transversal
adiabatic transformations and active correction could provide a
solution to the problem of fault tolerance in the adiabatic model of
computation.

\acknowledgements{O.O. acknowledges the support of the European Commission under the Marie Curie Intra-European Fellowship Programme (PIEF-GA-2010-273119).  This research was supported in part by the ARO MURI grant W911NF-11-1-0268, and by the Spanish
MICINN (Consolider-Ingenio QOIT). }

\bibliographystyle{plain}
\bibliography{refs}

\end{document}